\documentclass[fleqn,preprint]{elsarticle}

\usepackage{
amsmath,amssymb}
\usepackage[T1]{fontenc} 

\usepackage{tikz} 
\usepackage{verbatim}
\usetikzlibrary{calc,positioning,arrows}
\usepackage{paralist}
\usepackage[ruled,vlined] 
       {algorithm2e} 
\SetAlgoInsideSkip{medskip}

\usepackage{graphicx,xcolor}
\usepackage{float}
\usepackage{wrapfig}

\usepackage{hyperref}
\hypersetup{
  colorlinks   = true, 
  urlcolor     = blue, 
  linkcolor    = blue, 
  citecolor   = blue 
}

\newtheorem{theorem}{Theorem}

\newtheorem{observation}{Observation}
\newtheorem{problem}{Open Problem}
\newtheorem{definition}{Definition}

\newenvironment{proof}{\noindent {\bf Proof.\,}}{\qed\smallskip}

{\bf}{\it}
{\bf}{\it}

\newcommand{\C}{\mathcal{C}}

\newcommand{\N}{\mathbb{N}}

\newcommand{\NP}{\ensuremath{\mathsf{NP}}\xspace}

\newcommand{\coNpo}{\mathsf{coNP/poly}}
\newcommand{\containment}{\ensuremath{\mathsf{NP\subseteq coNP/poly}}\xspace}
\newcommand{\mktwo}{\ensuremath{\mathsf{MK[2]}}\xspace}

\newcommand{\vcfc}{\mathrm{vcfc}}
\newcommand{\svcfc}{\mathrm{svcfc}}

\newcommand{\dist}{\mathrm{dist}}

\newcommand{\Oh}{\mathcal{O}}

\def\PNAESAT{\textup{\textsc{positive nae sat}}}
\def\threeSAT{\textup{\textsc{$3$-sat}}}

\def\PNAEthreeSAT{\textup{\textsc{positive nae $3$-sat}}}

\newcommand{\CFVC}{\textup{\textsc{cfvc number}}}

\newcommand{\SCFVC}{\textup{\textsc{strong cfvc number}}}
\newcommand{\SCFVCthree}{\textup{\textsc{strong cfvc $3$-coloring}}}
\newcommand{\SCFVCtwo}{\textup{\textsc{strong cfvc $2$-coloring}}}
\newcommand{\SCFVCk}{\textup{\textsc{strong cfvc $k$-coloring}}}

\usepackage{lineno}

\tikzstyle{vertex}=[draw,circle,inner sep=1.3pt,fill=black]
\tikzstyle{ab}=[draw,circle,inner sep=1.3pt]
\tikzstyle{vertexV}=[draw, circle, minimum size=3mm, inner sep=.34pt] 
\tikzstyle{vertexC}=[draw, circle, inner sep=.35pt] 

\begin{document}

\begin{frontmatter}

\title{The parameterized complexity of Strong Conflict-Free Vertex-Connection Colorability}

\author[1]{Carl Feghali}
\ead{carl.feghali@ens-lyon.fr}

\author[2]{Hoang-Oanh Le}
\ead{HoangOanhLe@outlook.com}

\author[3]{Van Bang Le\corref{corr}}
\ead{van-bang.le@uni-rostock.de}

\cortext[corr]{Corresponding author}

\affiliation[1]{organization={Univ Lyon, EnsL, CNRS, LIP},
addressline={F-69342, Lyon Cedex 07},
country={France}}

\affiliation[2]{organization={Independent Researcher},
addressline={Berlin},
country={Germany}}

\affiliation[3]{organization={Universit\"at Rostock, Rostock},
country={Germany}}

\begin{abstract}
This paper continues the study of a new variant of graph coloring with a connectivity constraint recently introduced by Hsieh \textit{et~al.} [COCOON 2024].
A path in a vertex-colored graph is called \emph{conflict-free} if there is a color that appears exactly once on its vertices. 
A connected graph is said to be \emph{strongly conflict-free vertex-connection $k$-colorable} if it admits a (proper) vertex $k$-coloring such that any two distinct vertices are connected by a  conflict-free \emph{shortest} path. 
Among others, we show that deciding, for a given graph $G$ and an integer $k$, whether~$G$ is strongly conflict-free $k$-colorable is fixed-parameter tractable when  parameterized by the vertex cover number. But under the standard complexity-theoretic assumption $\NP\not\subseteq\coNpo$, deciding, for a given graph $G$, whether~$G$ is strongly conflict-free $3$-colorable does not admit a polynomial kernel, even for bipartite graphs. This kernel lower bound is in stark contrast to the ordinal {\sc $k$-coloring} problem which is known to admit a polynomial kernel when parameterized by the vertex cover number.
\end{abstract}


\end{frontmatter}

\section{Introduction and results}
Let $G$ be a vertex-colored graph. We say that a path in $G$ is \emph{conflict-free} if there is a color that occurs exactly once on the path. 
A graph is said to be \emph{conflict-free vertex-connection $k$-colorable}, \emph{cfvc $k$-colorable} for short, if it admits a (not necessarily proper) vertex $k$-coloring such that any two distinct vertices are connected by a conflict-free path. The \emph{conflict-free vertex-connection number} of a graph~$G$, denoted $\vcfc(G)$, is the smallest integer $k$ such that $G$ is cfvc $k$-colorable. The concept of cfvc $k$-colorability has been introduced in~2020 in~\cite{LiZZMZJ20}. Since then, many research papers have focused on this topic; see, e.g.,~\cite{BrauseJS18,DoanHS22,DoanS21,JiLZ20,LiZ20,LiW18+} and the recent survey~\cite{ChangH24}.  

Recently, the paper~\cite{HsiehLLP24} initiated the study of a strong version of cfvc $k$-colorability by requiring that any two distinct vertices must be connected by a conflict-free \emph{shortest} path.  

\begin{definition}[\cite{HsiehLLP24}, strong cfvc colorings] 
Let~$G$ be a graph and $k\ge 1$ be an integer. $G$ is \emph{strongly conflict-free vertex-connection $k$-colorable}, \emph{strong cfvc $k$-colorable} for short, if there is a function $f:V(G)\to [k]=\{1,2,\ldots,k\}$ such that any two distinct vertices~$u$ and~$v$ of~$G$ are connected by a \emph{conflict-free shortest $u,v$-path}, a shortest $u,v$-path on which there is a color $c\in [k]$ that occurs exactly once.

\noindent
The \emph{strong conflict-free vertex-connection number} of $G$, \emph{strong cfvc number} for short, denoted $\svcfc(G)$, is the smallest integer $k$ such that $G$ is strongly cfvc $k$-colorable.
\end{definition}

This paper deals with the problem of deciding, for a given graph $G$ and an integer $k$, whether $G$ admits a strong cfvc $k$-coloring: 

\medskip\noindent
\fbox{
\begin{minipage}{.96\textwidth}
\SCFVC\\[.7ex]%
\begin{tabular}{l l}
{\em Instance:}& A connected graph $G$ and an integer $k$.\\
{\em Question:}& $\svcfc(G)\le k$, i.e., is $G$ strongly cfvc $k$-colorable\,?
\end{tabular}
\end{minipage}
}

\medskip\noindent
When $k$ is a given constant, i.e., $k$ is not part of the input, we write \SCFVCk\ instead of  \SCFVC. For the sake of completeness, we explicitly give its description below. 

\medskip\noindent
\fbox{
\begin{minipage}{.96\textwidth}
\SCFVCk\\[.7ex]%
\begin{tabular}{l l}
{\em Instance:}& A connected graph $G$.\\
{\em Question:}& $\svcfc(G)\le k$, i.e., is $G$ strongly cfvc $k$-colorable\,?
\end{tabular}
\end{minipage}
}

\medskip\noindent
It should be noted that \SCFVC\ is generally more difficult than \SCFVCk; any efficient algorithm for \SCFVC\ must have a runtime that is polynomial in the input size, in particular polynomial in $k$, while any efficient algorithm for \SCFVCk\ has a runtime that only needs to be polynomial in the size of input $G$ (but can be exponential in $k$).

\medskip\noindent
Note that, unlike cfvc colorings, strong cfvc colorings are proper colorings.  
Further, unlike proper colorings, strong cfvc colorings are \emph{not} monotone: deleting vertices may turn the instance from yes to no; see an example depicted in Fig.~\ref{fig:non-monotone}. 

\begin{figure}[H]
\centering
\begin{tikzpicture}[scale=.35]
\node[vertex] (v1) at (1,3) [label=left:{\small $1$}]{};
\node[vertex] (v2) at (3,5) [label=left:{\small $2$}]{};
\node[vertex] (v2') at (3,1) [label=left:{\small $3$}]{};
\node[vertex] (v3) at (5,3) [label=left:{\small $1$}]{};
\node[vertex] (v4) at (7,5) [label=right:{\small $u$}]{};
\node at (7.2,5) [label=left:{\small $2$}]{};
\node[vertex] (v4') at (7,1) [label=right:{\small $v$}]{};
\node at (7.2,1) [label=left:{\small $3$}]{};
\node[vertex] (v5) at (9,3) [label=above:{\small $1$}]{};
\node[vertex] (v6) at (11,3) [label=above:{\small $3$}]{};
\node[vertex] (v7) at (13,3) [label=above:{\small $2$}]{};
\node[vertex] (v8) at (15,3) [label=above:{\small $1$}]{};

\draw (v1)--(v2)--(v3)--(v4')--(v5)--(v6)--(v7)--(v8); 
\draw (v1)--(v2')--(v3)--(v4)--(v5);

\end{tikzpicture} 
\qquad\quad
\begin{tikzpicture}[scale=.35]
\node[vertex] (v1) at (1,3) {};
\node[vertex] (v2) at (3,5) {};
\node[vertex] (v2') at (3,1) {};
\node[vertex] (v3) at (5,3) {};
\node[vertex] (v4') at (7,1) [label=right:{\small $v$}]{};
\node[vertex] (v5) at (9,3) {};
\node[vertex] (v6) at (11,3) {};
\node[vertex] (v7) at (13,3) {};
\node[vertex] (v8) at (15,3) {};

\draw (v1)--(v2)--(v3)--(v4')--(v5)--(v6)--(v7)--(v8); 
\draw (v1)--(v2')--(v3); 
\end{tikzpicture} 
\caption{Left: a graph $G$ with a strong conflict-free $3$-coloring and two twins $u$ and $v$. 
Right: it can be verified that $G-u$ is \emph{not} strongly conflict-free $3$-colorable. 
}\label{fig:non-monotone}
\end{figure}
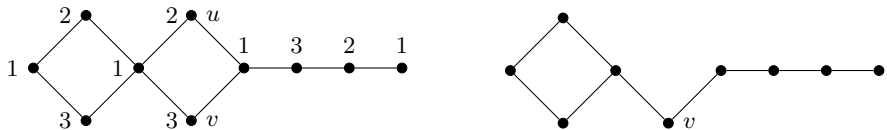

This fact indicates, as discussed in~\cite{HsiehLLP24}, that determining the complexity of \SCFVC\ and \SCFVCk\ may be harder than that of the classical coloring problems {\sc chromatic number} and $k$-{\sc coloring}, respectively. 

\paragraph{Previous work} 
Since paths in trees are unique, the concepts of cfvc colorings and strong cfvc colorings coincide when restricted to trees. Thus, strong cfvc colorings in trees have been investigated implicitly in~\cite{LiZZMZJ20,LiZ20}. In~\cite{LiZZMZJ20}, Li et al. realized that determining the strong cfvc number of trees is \lq\lq very difficult.\rq\rq\ 
For $n$-vertex paths $P_n$, they proved that $\svcfc(P_n)=\lceil\log_2(n+1)\rceil$. Several upper bounds for the (strong) cfvc number of trees are given in~\cite{LiZZMZJ20}. It follows from~\cite{LiW18+} that $\svcfc(T)$ is upper bounded by $\lceil\log_2(n+1)\rceil$ for $n$-vertex trees~$T$. Exact value of $\svcfc(T)$ for trees $T$ of diameter at most~$4$ has been determined in~\cite{LiZ20}. Whether or not the cfvc number of a given tree can be computed in polynomial time is still an open problem; see~\cite{Li25,LiW18+,Wu25}. 

The concept of strong cfvc coloring is formally introduced in~\cite{HsiehLLP24}. 
The authors of~\cite{HsiehLLP24} observed that strongly cfvc $2$-colorable graphs are exactly the complete bipartite graphs, hence, \SCFVCtwo\ can be solved in linear time. They proved that \SCFVCk\ is \NP-complete for diameter-$d$ graphs for all pairs $(k,d)$ with $k\ge 3$ and $d\ge 2$ except when $(k,d)=(3,2)$. Moreover, \SCFVCthree\ remains \NP-complete even when restricted to $3$-colorable graphs with diameter~$3$, radius~$2$ and domination number~$3$, and cannot be solved in subexponential time in the vertex number on these restricted graphs unless the Exponential Time Hypothesis fails. When restricted to split graphs and complement of bipartite graphs, \SCFVC\ can be solved in linear time. 

Observe that a graph has diameter at most~$2$ if and only if every two non-adjacent vertices have a common neighbor. Hence, in graphs of diameter at most~$2$, the two concepts, proper colorings and strong cfvc colorings,  coincide. Thus, strong cfvc colorings have been studied implicitly in the literature: we remark that, while it is well known and easy to see that in diameter-$2$ graphs, $k$-{\sc coloring}, hence \SCFVCk, is \NP-complete for any fixed $k\ge4$, determining the complexity of $3$-{\sc coloring} in diameter-$2$ graphs, hence the complexity of \SCFVCthree, in diameter-$2$ graphs, is a notoriously difficult, well-known long-standing open problem in algorithmic graph theory (cf.~\cite{MertziosS16,DebskiPR22}). 

Concerning parameterized complexity, Cai~\cite[Theorem 4.4]{Cai03a} proved that \textsc{chromatic number} parameterized by the distance to a split graph is $\mathsf{W}[1]$-hard by giving a parameterized reduction from \textsc{independent set}. Since the constructed graph instance for \textsc{chromatic number} has diameter~$2$, Cai~\cite{Cai03a} therefore proved  that \SCFVC\ parameterized by the distance to a split graph is $\mathsf{W}[1]$-hard. 

\paragraph{Our results} 
In this paper, we focus on the parameterized complexity of \SCFVC\ and \SCFVCk. Since \SCFVC\ is para-hard for the \lq\lq natural\rq\rq\ parameter $k$ (the problem is \NP-complete for any fixed $k\ge3$), we parameterize \SCFVC\ by the most popular structural parameter, the vertex cover number. The contributions of this paper are as follows. 

{\parindent5mm
\begin{itemize}
\item \SCFVC\ parameterized by the vertex cover number is fixed-parameter tractable.
\item Unless $\containment$, \SCFVCthree\ parameterized by the vertex cover number does \emph{not} admit a polynomial kernel, even when restricted to bipartite graphs of bounded diameter.
\end{itemize}
}

\noindent
Recall that strong cfvc colorings are proper colorings. Relating to the classical \textsc{$k$-coloring} problem, we remark that the kernel lower bound for \SCFVCthree\ parameterized by the vertex cover is in stark contrast to \textsc{$k$-coloring}. It is known (\cite{JansenK13,JansenP19}) that \textsc{$k$-coloring} parameterized by the vertex cover does admit a polynomial kernel. 
Also, we will point out that \SCFVCthree\ parameterized by the distance to a path does \emph{not} admit a polynomial kernel. 

\noindent
\paragraph{Organization} We provide preliminaries in section~\ref{sec:pre}. In section~\ref{sec:fpt} we prove that \SCFVC\ is fpt when parameterized by the vertex cover number. In section~\ref{sec:nokernel} we show that \SCFVCthree\ parameterized by the vertex cover number and by the distance to a path admits no polynomial kernel unless $\containment$. Section~\ref{sec:con} concludes the paper with some open questions for further research. 

\section{Preliminaries}\label{sec:pre}
We consider only finite, simple and \emph{connected} undirected graphs $G=(V,E)$ with vertex set $V(G)=V$ and edge set $E(G)=E$. 
The \emph{distance} between two vertices $u$ and $v$, written $\dist(u,v)$, is the length of a shortest $u,v$-path. 
The \emph{diameter} of~$G$ is the maximum distance between any two vertices in~$G$. 
We use $N(v)$ to denote the neighborhood of a vertex $v$ and, $N[v]=N(v)\cup\{v\}$ for its closed neighborhood.
Two vertices $u$ and $v$ of $G$ are \emph{true twins} if $N[u]=N[v]$ and \emph{false twins} if $N(u)=N(v)$. 
Say that $u$ and $v$ are \emph{twins} if they are true twins or false twins.

\paragraph{Parameterized complexity and kernelization}
A \emph{parameterized problem} is a language $L\subseteq\Sigma^*\times\N$; the second component $k$ of an instance $(x,k)$ is called its \emph{parameter}. A parameterized problem $L$ is \emph{fixed-parameter tractable}, \emph{fpt} for short, if it can be decided whether $(x,k)\in L$ in time $f(k)|x|^c$ for a computable function $f\colon\N\to\N$ and a constant $c$. A \emph{kernelization}, or just \emph{kernel}, for a parameterized problem $L$ is an efficient algorithm that given $(x,k)\in\Sigma^*\times\N$ takes time polynomial in $|x|+k$ and returns an instance $(x',k')$ of size $|x'|+k'\le h(k)$ for some computable function $h\colon\N\to\N$ and with $(x,k)\in L$ if and only if $(x',k')\in L$. The function $h$ is also called the \emph{size} of the kernel and for a \emph{polynomial kernel} it is required that $h$ is polynomially bounded. It is well known that a parameterized problem is fpt if and only if it has a kernel. 
A \emph{polynomial-parameter transformation} from $Q\subseteq\Sigma^*\times\N$ to $Q'\subseteq\Gamma^*\times\N$ is an efficient algorithm which takes an instance  $(x,k)$ of $Q$ as input and outputs, in time polynomial in $|x|+k$, an equivalent instance $(x',k')$ of $Q'$ with $k'$ is polynomially bounded in~$k$: $(x,k)\in Q$ if and only if $(x',k')\in Q'$ and $k'=k^{\Oh(1)}$.

\paragraph{Graph parameters} Let $G=(V,E)$ be a graph. A \emph{vertex cover} of $G$ is a set $X\subseteq V$ such that $G-X$ is an independent set. 
For a graph class $\C$, a \emph{modulator of $G$ to $\C$} is a set $X\subseteq V$ such that $G-X\in\C$. The \emph{distance} to $\C$ is the smallest size of a modulator of $G$ to $\C$. E.g., the vertex covers are exactly the modulators and the vertex cover number is the distance to the class of edgeless graphs. 

\section{\SCFVC\ parameterized by the vertex cover number}\label{sec:fpt}
In this section, we give an fpt-algorithm solving \SCFVC\ parameterized by the vertex cover number.

\medskip\noindent
\fbox{
\begin{minipage}{.96\textwidth}
\SCFVC$[\textsc{vc}]$\\[.7ex]
\begin{tabular}{l l}
{\em Instance:}& A connected graph $G$, an integer $k\in\N$ and a vertex\\ 
                      & cover $X\subseteq V(G)$ of $G$.\\
{\em Parameter:}& $|X|$.\\
{\em Question:}& $\svcfc(G)\le k$, i.e., is $G$ strongly cfvc $k$-colorable\,?
\end{tabular}
\end{minipage}
}

\medskip\noindent
Note that in \SCFVC$[\textsc{vc}]$ we may assume that a vertex cover is given together with the input graph because it is well known that a $2$-approximate vertex cover can be found in polynomial time. 

We need an observation about twins. By the example depicted in Fig.~\ref{fig:non-monotone}, \CFVC\ and \SCFVC\ are not monotone, even when deleting one of some false twins. However, some useful facts about deleting twins do hold, that will simplify some arguments in our proof below.

\begin{observation}\label{obse:twins}
Let $u$ and $v$ be two false twins in a connected graph $G$.
\begin{itemize}
\item[\em (i)] If $G-u$ is strongly cfvc $k$-colorable, then $G$ is strongly cfvc  $k$-colorable. 
\item[\em (ii)] If $f$ is a strong cfvc $k$-coloring of $G$ with $f(u)=f(v)$, then the restriction of $f$ on $G-u$ is a strong cfvc $k$-coloring of $G-u$. 
\end{itemize}
\end{observation}
\begin{proof} 
(i):  Let $f$ be a strong cfvc $k$-coloring of $G-u$. Extend $f$ to a proper $k$-coloring of $G$ by setting $f(u):= f(v)$. Then $f$ is indeed a strong cfvc $k$-coloring of $G$: since $u$ and $v$ are twins in $G$, for any vertex $x\not=u$, if $P$ is a shortest $x,v$-path in $G-u$ then $P-v+u$ is a shortest $x,u$-path in $G$. Since there exists some conflict-free shortest $x,v$-path in $G-u$, some conflict-free shortest $x,u$-path exists in $G$. 
By the same reason, for any distinct vertices $x,y\in V(G)\setminus\{u,v\}$, any conflict-free shortest $x,y$-path in $G-u$ is a conflict-free shortest $x,y$-path in $G$. (Note that adding $u$ cannot decrease the distance for any pair $x,y$.) 

(ii): 
We will see that, under the restriction of $f$ on $G-u$, any two vertices in $G-u$ are connected by a conflict-free shortest path in $G-u$. 
Let $x\not= y$ be two vertices in $G-u$, and let $P$ be a conflict-free shortest $x,y$-path in $G$. 
If~$u$ is not contained in $P$, then $P$ is a conflict-free shortest $x,y$-path in $G-u$. If~$P$ contains~$u$ but not~$v$, then $P'=P-u+v$ is a shortest $x,y$-path in $G-u$, and~$P'$ is conflict-free because $f(u)=f(v)$. Finally, if~$P$ contains both~$u$ and~$v$, then $P=uwv$ for some $w\in N(u)=N(v)$. That is, $\{x,y\}=\{v,w\}$ and $P=xy$ is a conflict-free shortest $x,y$-path in $G-u$: $f$ is particularly a proper coloring, hence $f(x)\not=f(y)$.  
\end{proof}

\begin{theorem}\label{thm:vc-fpt}
\SCFVC$[\textup{\textsc{vc}}]$ can be solved in time $t(|X|)\cdot\Oh(n^5)$, where $t(|X|)=|X|^{|X|+(|X|+1)\cdot 2^{|X|}}$ and $n$ is the number of vertices of the input graph.
\end{theorem}
\begin{proof} 
Let $(G,k,X)$ be an instance of \SCFVC$[\textsc{vc}]$. Note that $I=V(G)\setminus X$ is an independent set and every vertex $v\in I$ has a neighbor in $X$. We may assume that $k\le |X|$, otherwise $G$ is clearly strongly cfvc $k$-colorable: color each vertex in $X$ with a unique color from $\{1,2,\ldots,|X|\}$ and color all vertices in~$I$ with the same color $|X|+1$. 

We first apply the following reduction rule exhaustively to obtain an equivalent instance $(G',k',X')$ of at most $|X|+(|X|+1)\cdot 2^{|X|}$ vertices. 
For each subset $S\subseteq X$,  let $I_S$ be set of all (false) twins in~$I$ that are adjacent to exactly the vertices in $S$, $I_S=\{v\in I\mid N(v)=S\}$:

\smallskip\noindent
\colorbox{gray!20!white}{
\begin{minipage}{.96\textwidth}
\textbf{Reduction Rule.}\, 
If, for some $S\subseteq X$, $|I_S|>k+1$ then remove a vertex $u$ from $I_S$, that is, output $G'=G-u$, $k'=k$ and $X'=X$. 
\end{minipage}
}

\smallskip\noindent
For the safeness of the reduction rule, we will see that $G$ is strongly cfvc $k$-colorable if and only if $G-u$ is strongly cfvc $k$-colorable. Let $S\subseteq X$ with $|I_S|>k+1$, and let $J=I_S\setminus\{u\}$; notice that $I_S$ is a set of false twins and that $|J|\ge k+1$. 
Now, if $G-u$ is strongly cfvc $k$-colorable, then by Observation~\ref{obse:twins}~(i), $G$ is strongly cfvc $k$-colorable.  
Conversely, assume that $G$ admits a strong cfvc $k$-coloring $f$. If $f(u)\in f(J)$ then, by Observation~\ref{obse:twins}~(ii), the restriction of $f$ on $G-u$ is a strong cfvc $k$-coloring of $G-u$. Otherwise, $f(u)\not\in f(J)$. As $|J|\ge k+1$, there are two distinct vertices $v,w\in J$ with $f(v)=f(w)$. Define a (proper) $k$-coloring~$f'$ of~$G$ as follows: $f'(z)=f(z)$ for all $z\not=v$ and set $f'(v)=f(u)$. Then, as~$u$ and~$v$ are twins in $G$, $f'$ is a strong cfvc $k$-coloring of~$G$, with $f'(u)=f'(v)$. By Observation~\ref{obse:twins}~(ii) again, the restriction of~$f'$ on $G-u$ is a strong cfvc $k$-coloring of~$G-u$.  

We have seen that the reduction rule is safe. Observe that the reduced equivalent instance can be obtained in fpt-time $2^{|X|}\cdot\Oh(n^3)$: the reduced rule can be applied at most $|I|-(k+1)<n$ times in total. Checking, if there is some $S\subseteq X$ for which the reduction rule is applicable, can be done in time $2^{|X|}\cdot\Oh(n^2)$.

If the reduction rule is not longer applicable, then for any $S\subseteq X$, there exist at most $k+1$ vertices $u\in I$ with $N(u)=S$. Hence, since there are $2^{|X|}$ subsets $S\subseteq X$, the reduced equivalent instance $(G',k,X)$ has at most $|X|+(k+1)\cdot 2^{|X|}\le |X|+(|X|+1)\cdot 2^{|X|}$ vertices. Hence, deciding whether $(G,k,X)$ is an yes-instance can be done by checking if some proper $k$-coloring of $(G',k,X)$ is indeed a strong cfvc $k$-coloring. Since it can be verified in time $\Oh(kn^4)=\Oh(n^5)$ if a $k$-coloring is a strong cfvc $k$-coloring (cf.~\cite{HsiehLLP24}), one can decide whether $(G,k,X)$ is an yes-instance in time $2^{|X|}\cdot\Oh(n^3)+k^{|X|+(k+1)\cdot 2^{|X|}}\cdot\Oh(n^5)$ which is roughly bounded by $|X|^{|X|+(|X|+1)\cdot 2^{|X|}}\Oh(n^5)$. 
\end{proof}

\section{A reduction from \PNAESAT\ to \SCFVCthree\ and its consequences}\label{sec:nokernel}
In this section, we first give a polynomial reduction from the \NP-complete problem \textsc{positive not-all-equal sat} (\cite{Schaefer78}), \PNAESAT\ for short, or equivalently, \textsc{set splitting} or \textsc{hypergraph $2$-coloring}, to \SCFVCthree, restricted to bipartite graphs of bounded diameter. Then, we will derive from this reduction a lower bound based on the Exponential Time Hypothesis for exact algorithms solving \SCFVCthree, as well as a kernel lower bound based on the standard assumption $\NP\not\subseteq\coNpo$ for \SCFVCthree\ parameterized by the vertex cover number.

Recall that an instance for \PNAESAT\ is a Boolean formula in conjunctive normal form (\emph{CNF formula} for short) $\C$, a collection of clauses over a variable set~$U$, in which no negative literal appears. The problem asks whether there is a truth assignment of the variables such that every clause in $\C$ has at least one true and at least one false variable. Such an assignment is called an \emph{not-all-equal-assignment}, \emph{NAE-assignment} for short. 
Let $(U,\mathcal{C})$ be an instance of \PNAESAT\ over variable set $U$ 
and clause set $\mathcal{C}$. 
If some clause $C \in \mathcal{C}$ has size at most one then the instance is trivially no. Therefore, we may assume that every clause contains at least two variables. 
We construct an instance $G$ for \SCFVCthree\ as follows. 

\begin{itemize}
\item We begin by taking a chain of six $4$-cycles, consisting of 19 vertices $a_1, a_2,\allowbreak\ldots,\allowbreak a_6$, $b_1, b_2,\allowbreak\ldots,\allowbreak b_6$ and $h_1,h_2,\allowbreak\ldots,\allowbreak h_7$,  along with $24$ edges $h_ia_i$, $a_ih_{i+1}$, $h_{i+1}b_i$ and $b_ih_i$, $1\le i\le 6$.  
\item We take two additional vertices $a_7$ and $b_7$, and edges $h_7a_7$ and $h_7b_7$.
\item For each clause $C\in \mathcal{C}$ we have a \emph{clause vertex}~$c$.
\item For each variable $v\in U$ we have a \emph{variable vertex}~$v$.
\item We add edges between clause vertices $c$ and variable vertices $v$ if $v\in C$.
\item We add edges between~$a_7$ and all clause vertices, and between~$b_7$ and all clause vertices.
\item Finally, we add edges between~$h_1$ and all variable vertices to obtain the graph~$G$. 
\end{itemize}
Observe that $G$ has $|U|+|\C|+21$ vertices. See Fig.~\ref{fig:G} for an example.

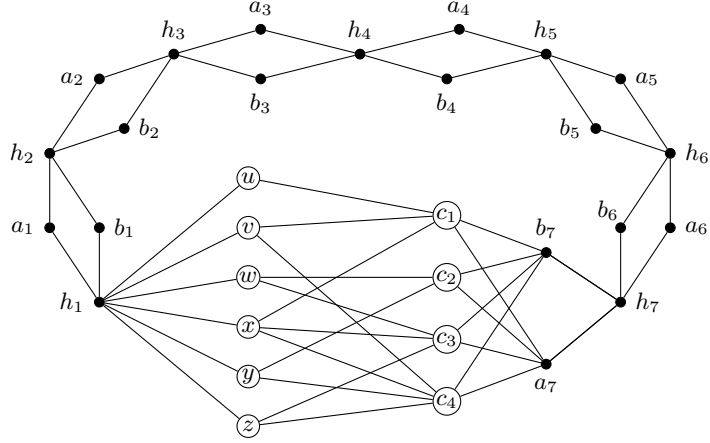
\begin{figure}[ht]
\centering
\begin{tikzpicture}[scale=.33]
\node[vertexV] (u) at (9,11) {\small $u$};
\node[vertexV] (v) at (9,9) {\small $v$};
\node[vertexV] (w) at (9,7) {\small $w$};
\node[vertexV] (x) at (9,5) {\small $x$};
\node[vertexV] (y) at (9,3) {\small $y$};
\node[vertexV] (z) at (9,1) {\small $z$};
\node[vertexC] (c1) at (17,9.5) {\small $c_1$};
\node[vertexC] (c2) at (17,7) {\small $c_2$};
\node[vertexC] (c3) at (17,4.5) {\small $c_3$};
\node[vertexC] (c4) at (17,2) {\small $c_4$};

\draw (c1)--(u); \draw (c1)--(v); \draw (c1)--(x);
\draw (c2)--(w); \draw (c2)--(y);
\draw (c3)--(w); \draw (c3)--(x); \draw (c3)--(z);
\draw (c4)--(v); \draw (c4)--(x); \draw (c4)--(y); \draw (c4)--(z);

\node[vertex] (h1) at (3,6) [label=left:{\small $h_1$}]{};
\node[vertex] (h2) at (1,12) [label=left:{\small $h_2$}]{};
\node[vertex] (h3) at (6,16) [label=above:{\small $h_3$}]{};
\node[vertex] (h4) at (13.5,16)  [label=above:{\small $h_4$}]{};
\node[vertex] (h5) at (21,16)  [label=above:{\small $h_5$}]{};
\node[vertex] (h6) at (26,12)  [label=right:{\small $h_6$}]{};
\node[vertex] (h7) at (24,6)  [label=right:{\small $h_7$}]{};
\node[vertex] (a1) at (1,9) [label=left:{\small $a_1$}]{};
\node[vertex] (a2) at (3,15) [label=left:{\small $a_2$}]{};
\node[vertex] (a3) at (9.5,17) [label=above:{\small $a_3$}]{};
\node[vertex] (a4) at (17.5,17) [label=above:{\small $a_4$}]{};
\node[vertex] (a5) at (24,15) [label=right:{\small $a_5$}]{};
\node[vertex] (a6) at (26,9) [label=right:{\small $a_6$}]{};
\node[vertex] (a7) at (21,3.5) [label=below:{\small $a_{7}$}]{};

\node[vertex] (b1) at (3,9) [label=right:{\small $b_1$}]{};
\node[vertex] (b2) at (4,13) [label=right:{\small $b_2$}]{};
\node[vertex] (b3) at (9.5,15) [label=below:{\small $b_3$}]{};
\node[vertex] (b4) at (17,15) [label=below:{\small $b_4$}]{};
\node[vertex] (b5) at (23,13) [label=left:{\small $b_5$}]{};
\node[vertex] (b6) at (24,9) {}; 
\node at (23.5,9.7) {\small $b_6$};
\node[vertex] (b7) at (21,8) [label=above:{\small $b_7$}]{};

\draw (h1)--(u); \draw (h1)--(v); \draw (h1)--(w); \draw (h1)--(x); \draw (h1)--(y); \draw (h1)--(z);
\draw (a7)--(c1); \draw (a7)--(c2); \draw (a7)--(c3); \draw (a7)--(c4);
\draw (b7)--(c1); \draw (b7)--(c2); \draw (b7)--(c3); \draw (b7)--(c4);
\draw (h7)--(a7); \draw (h7)--(b7);

\draw (h1)--(a1)--(h2)--(a2)--(h3)--(a3)--(h4)--(a4)--(h5)--(a5)--(h6)--(a6)--(h7)--(a7);
\draw (h1)--(b1)--(h2)--(b2)--(h3)--(b3)--(h4)--(b4)--(h5)--(b5)--(h6)--(b6)--(h7)--(b7); 

\end{tikzpicture} 

\caption{The bipartite graph $G$ from the formula $(U,\C)$ over variable set $U=\{u,v,w,x,y,z\}$ and clause set $\C=\{C_1,C_2,C_3,C_4\}$ with $C_1=\{u, v, x\}$, $C_2=\{w,y\}$, $C_3=\{w,x,z\}$ and $C_4= \{v,x,y,z\}$.}\label{fig:G}
\end{figure}

Clearly the construction of $G$ can be carried out in polynomial time, and $G$ is bipartite and has diameter~$8$ and radius~$8$. 
It remains to verify correctness, i.e., that $(U,\mathcal{C})$ is yes for \PNAESAT\ if and only if $G$ is yes for \SCFVCthree. 

For one direction, assume that there is a satisfying NAE-assignment \texttt{b} for $\mathcal{C}$. Then color the vertices of $G$ as follows:
\begin{itemize}
\item color variable vertex $v$ by color~$1$ if $\texttt{b}(v)=\texttt{True}$. Otherwise by color~$0$;
\item color all clause vertices $c$ by color $2$;
\item color all $h_i$ by color $2$;
\item color all $a_i$ by color 0 and all $b_i$ by color $1$.
\end{itemize}

It is not hard to see that this proper $3$-coloring of $G$ is a strong cfvc $3$-coloring: for any two vertices~$x$ and~$y$ of distance at least~$3$, there is a shortest $x,y$-path that contains only one vertex colored by~$0$ or a shortest $x,y$-path that contains only one vertex colored by~$1$. 

For the other direction, assume that $G$ admits a strong cfvc $3$-coloring $f$ with colors $0,1$ and~$2$. We will see that $(U,\C)$ is yes for \PNAESAT. Without loss of generality, let 
\begin{align*}
f(h_1)=2.
\end{align*}
Then the variable vertices are colored by $0$ or $1$. We assign variable $v$ to \texttt{True} if $f(v)=1$ and to \texttt{False} if $f(v)=0$. Clearly, this assignment is well-defined. 
In order to see that this assignment NAE-satisfies $\mathcal{C}$, we show (for a clause $C\in\C$, by $C$-variable vertices we mean the variable vertices whose corresponding variables belong to clause~$C$):

\smallskip\noindent
\textsc{Claim:} {\em  For every clause $C$, at least one $C$-variable vertex is colored by~$0$ and at least one $C$-variable is colored by~$1$\/}.\\ 
\textsc{Proof of the claim:} Suppose to the contrary that there is a clause~$C\in\C$ such that all $C$-variable vertices are colored by the same color. Then we will encounter a contradiction, a \emph{bad pair}, where two vertices form a bad pair if there is no conflict-free shortest path in~$G$ connecting them. Since~$f$ is a strong conflict-free coloring of~$G$, there is no bad pair in~$G$. 

Without loss of generality, let all $C$-variable vertices be colored by~$0$. 
Then, the clause vertex $c$ is colored by~$1$ or~$2$. Moreover,
\begin{align}\label{fact1}
\text{$f(a_1)=1$ or $f(c)=1$, and $f(b_1)=1$ or $f(c)=1$.}
\end{align}
Otherwise, $a_1$ and $c$, or $b_1$ and $c$, would form a bad pair. (Note that every shortest $a_1,c$-path and every shortest $b_1,c$-path must go through $h_1$ and a $C$-variable vertex; all variable vertices, that are not $C$-variable vertices, are adjacent to $h_1$ but non-adjacent to~$c$.) 
For the same reason, for each~$i$, $2\le i\le 6$,  
\begin{align}\label{fact3}
\text{some vertices $x_i,y_i\in\{a_i,b_i,h_{i+1}\}$ have colors\,\,\quad}\\ \nonumber
\text{$f(x_i)\notin\{f(a_{i-1}),f(h_i)\}$, $f(y_i)\not\in\{f(b_{i-1}),f(h_i)\}$.}
\end{align}
Otherwise, $h_{i+1}$ and $a_{i-1}$, or $h_{i+1}$ and $b_{i-1}$, would form a bad pair. Notice that $x_i=y_i$ is possible, and that  
$f(\{a_{i-1},h_i,\allowbreak x_i\})\allowbreak = f\{b_{i-1},h_i,y_i\})\allowbreak =\{0,1,2\}$. It follows from (\ref{fact3}), that, for each~$i$, $2\le i\le 6$, 
\begin{align}\label{fact4}
\text{if $f(a_{i-1})\not=f(b_{i-1})$ and $f(h_i)=f(h_{i+1})$, then $f(a_i)\not=f(b_i)$.}
\end{align}
This is because, by (\ref{fact3}), if $f(a_{i-1})\not=f(b_{i-1})$ and $f(h_i)=f(h_{i+1})$ then $f(x_i)\not= f(y_i)$ and $\{x_i,y_i\}=\{a_i,b_i\}$. 

We now show that, for each $2\le i\le 6$, 
\begin{align}\label{fact5}
\text{$f(h_i)=2$ and $f(a_{i-1})\not=f(b_{i-1})$.}
\end{align}
Proof of (\ref{fact5}): we first verify the cases $i=2,3$. 
If $f(h_2)\not=2$, then $f(a_1)=f(b_1)$. Hence $f(c)=2$ (otherwise, by (\ref{fact3}), $x_2$ and~$c$ would form a bad pair because $f(\{x_2,h_2,a_1\})\allowbreak =f(\{x_2,h_2,b_1\})\allowbreak = f(\{h_1,v,c\})\allowbreak =\{0,1,2\}$ for any $C$-clause vertex $v$), and by (\ref{fact1}), $f(a_1)=f(b_1)=1$. 
Hence $f(h_2)=0$. Then no vertex $z\in\{a_2,b_2,h_3,a_3\}$ is colored by~$1$ (otherwise, $z$ and~$c$ would form a bad pair). But then $a_3$ and $h_2$ is a bad pair. Thus, $f(h_2)=2$, as claimed.

Now, if $f(a_1)=f(b_1)$ then $f(a_1)=f(b_1)=1$ (otherwise, $h_2$ and any $C$-variable vertex would form a bad pair), and $f(a_2)=f(b_2)=0$ (otherwise, $a_2$ and $h_1$, or $b_2$ and $h_1$, would form a bad pair). But then, if some $z\in\{a_3,h_3\}$ is colored by~$1$ then $z$ and any $C$-variable vertex form a bad pair, and if none of~$a_3$ and~$h_3$ is colored by~$1$ then $a_3$ and $h_2$ is a bad pair. Thus, $f(a_1)\not=f(b_1)$, as claimed. 
Then (\ref{fact1}) implies that
\begin{align*}
f(c)=1.
\end{align*}
Next, we verify the case $i=3$. If $f(h_3)\not=2$ then $f(a_2)=f(b_2)$ and $\{f(h_3),f(a_2)\}\allowbreak=\{f(h_3),f(b_2)\}=\{0,1\}$. But then $h_3$ and $c$ is a bad pair. Thus, $f(h_3)=2$, and by (\ref{fact4}), $f(a_2)\not=f(b_2)$. 

Finally, let $i\le 6$ be the largest index such that (\ref{fact5}) holds for all indices $j$ up to $i$, that is, $f(h_j)=2$ and $f(a_{j-1})\not=f(b_{j-1})$ for all $2\le j\le i$. We have seen that $i\ge 3$. Let $i<6$ (otherwise we are done), we show that~(\ref{fact5}) holds for $i+1$: if $f(h_{i+1})\not=2$ then $f(a_{i})=f(b_{i})$, and $\{f(h_{i+1}),f(a_i)\}=\{f(h_{i+1}),f(b_i)\}=\{0,1\}$. Now, by (\ref{fact3}), $x_{i+1}\in\{h_{i+2},a_{i+1},b_{i+1}\}$ has $f(x_{i+1})\in\{0,1\}$. Then $x_{i+1}$ and $a_{i-2}$, or $x_{i+1}$ and $b_{i-2}$ is a bad pair; note that, in~$G$, any shortest $x_{i+1},a_{i-2}$-path and any shortest $x_{i+1},b_{i-2}$-path must go through $h_{i+1},h_i$ and $h_{i-1}$. Thus, $f(h_{i+1})=2$, and by (\ref{fact4}), $f(a_i)\not=f(b_i)$. We have seen that (\ref{fact5}) holds for $i+1$, hence it holds for all $2\le i\le 6$. 

We now distinguish three cases: 
Let $f(h_7)=0$. Then $f(a_7)=f(b_7)=2$ and $f(a_6)=f(b_6)=1$. But then $c$ and $a_5$, or $c$ and $b_5$, is a bad pair. 

Let $f(h_7)=1$. Then $f(a_6)=f(b_6)=0$, and at least one of $a_7$ and $b_7$ is colored by~$2$ (otherwise, $c$ and $a_6$ would form a bad pair). 
Therefore, no clause vertex is colored by~$2$ because all clause vertices are adjacent to $a_7$ and $b_7$. Hence all edges joining a variable vertex and a clause vertex have an end colored by~$0$ and the other end colored by~$1$. But then $h_1$ and $a_6$ is a bad pair.

Finally, let $f(h_7)=2$. Then $f(a_7)=f(b_7)=0$. Moreover, by (\ref{fact4}), $f(a_6)\not=f(b_6)$. Hence no clause vertex is colored by~$2$: if $c'$ is a clause vertex with $f(c')=2$ then $c'$ and $a_6$, or $c'$ and $b_6$, is a bad pair. 
But then $h_1$ and $a_6$, or $h_1$ and $b_6$, is a bad pair. This final contradiction completes the proof of the Claim.

\smallskip\noindent
Now, for every clause $C$, by the claim, some variable in $C$ is assigned to \texttt{False} and some variable in~$C$ is assigned to \texttt{True}. That is, the assignment NAE-satisfies $\C$.

We have seen that the reduction is correct. Hence we obtain:

\begin{theorem}\label{thm:npc}
\SCFVCthree\ is \NP-complete, even when restricted to bipartite graphs of diameter~$8$ and radius~$8$.  
\end{theorem}

We now derive conditional lower bounds for exact exponential algorithms solving \SCFVCthree\ and for the kernel size of \SCFVCthree\ parameterized by the vertex cover number.

\subsection*{Exact exponential algorithms}
Recall first that there is a polynomial-time algorithm verifying whether a given proper $k$-coloring is strong conflict-free (\cite{HsiehLLP24}). Thus, by considering the at most $k^n$ proper $k$-colorings of an $n$-vertex graph, it follows that \SCFVCk\ can be solved in time $\Oh^*(k^{n})$. (The $\Oh^*$ notation ignores polynomial factors.)

Now, observe that our reduction from \PNAESAT\ to \SCFVCthree\ also works from \PNAEthreeSAT. Thus, we can derive a worst case lower bound on the running time for exact exponential algorithms solving \SCFVCk\ based on the well-known Exponential Time Hypothesis (ETH) (\cite{ImpagliazzoP01,ImpagliazzoPZ01}). The ETH, if true, implies that \threeSAT\ cannot be solved in time $2^{o(N+M)}$ for input CNF formulas over $N$ variables and $M$ clauses. 
This in turn implies that, assuming ETH, \PNAEthreeSAT\ cannot be solved in time $2^{o(N+M)}$, too. (It is well-known that there is a polynomial reduction from \threeSAT\ to \PNAEthreeSAT\ linear in $N$ and $M$.) 
Since the constructed graph~$G$ has $n=N+M +21$ vertices, where~$N$ and~$M$ is the number of variables and clauses, respectively, we obtain:

\begin{theorem}\label{thm:eth}
There is an algorithm solving \SCFVCk\ in time $\Oh^*(k^n)$ for $n$-vertex graphs. 
Moreover, assuming ETH, no algorithm can solve \SCFVCthree\ in time $2^{o(n)}$, even when restricted to $n$-vertex bipartite graphs of diameter~$8$ and radius~$8$.  
\end{theorem}

\subsection*{Kernel lower bound}
We now show that \SCFVCthree\ parameterized by the vertex cover number,  \SCFVCthree$[\textsc{vc}]$, has no polynomial kernel, unless \containment. 

\medskip\noindent
\fbox{
\begin{minipage}{.96\textwidth}
\SCFVCthree$[\textsc{vc}]$\\[.7ex]
\begin{tabular}{l l}
{\em Instance:}& A connected graph $G$ and a vertex cover $X\subset V(G)$ of $G$.\\
{\em Parameter:}& $|X|$.\\
{\em Question:}& $\svcfc(G)\le 3$, i.e., is $G$ strongly cfvc $3$-colorable\,?
\end{tabular}
\end{minipage}
}

\medskip\noindent
It can be observed that our reduction is a polynomial parameter transformation from \PNAESAT$[n]$, \emph{with parameter} $n$, the number of variables, to \SCFVCthree$[\textsc{vc}]$: let $X = U\cup \{h_1,h_2,\ldots,h_7\}\cup\{a_7,b_7\}$. Then~$X$ is a vertex cover of $G$, and $|X|=n+9$ is polynomial in the input parameter. 

It is known (\cite{KratschL24}) that \PNAESAT$[n]$ (aka \textsc{set splitting}$[n]$) is complete for the class called $\mktwo$ introduced in~\cite{HermelinKSWW15} under polynomial parameter transformations. Thus, \SCFVCthree$[\textsc{vc}]$ is $\mktwo$-hard, and therefore has no polynomial kernel unless $\containment$. (As proved in~\cite{HermelinKSWW15}, $\mktwo$-hard problems do not admit polynomial kernel under the assumption $\NP\nsubseteq\coNpo$.) 
Hence we obtain:

\begin{theorem}\label{thm:nopoly-vc-kernel}
Assuming $\NP\nsubseteq\coNpo$, \SCFVCthree$[\textup{\textsc{vc}}]$ 
 admits no polynomial kernel.
\end{theorem}

\smallskip
Finally, we point out that our polynomial-parameter transformation from \PNAESAT$[n]$ to \SCFVCthree$[\textsc{vc}]$ is easily modifiable to a polynomial-parameter transformation from \PNAESAT$[n]$ to \SCFVCthree$[\textsc{dp}]$, where \textsc{dp} is the \emph{distance to a path}. 

\medskip\noindent
\fbox{
\begin{minipage}{.96\textwidth}
\SCFVCthree$[\textsc{dp}]$\\[.7ex]
\begin{tabular}{l l}
{\em Instance:}& A connected graph $G$ and $X\subset V(G)$ such that $G-X$\\
                      & is a path.\\
{\em Parameter:}& $|X|$.\\
{\em Question:}& $\svcfc(G)\le 3$, i.e., is $G$ strongly cfvc $3$-colorable\,?
\end{tabular}
\end{minipage}
}

\medskip
\begin{theorem}\label{thm:no-poly-dp-kernel}
Assuming $\mathsf{NP\nsubseteq coNP/poly}$, \SCFVCthree$[\textup{\textsc{dp}}]$ 
 admits no polynomial kernel.
\end{theorem}
\noindent
\textbf{Proof Sketch.}
The reduction from \PNAESAT$[n]$ to \SCFVCthree$[\textup{\textsc{dp}}]$ is almost the same as the reduction from \PNAESAT$[n]$ to \SCFVCthree$[\textsc{vc}]$ at the beginning of this section. 

Given a CNF formula $(U,\C)$ with $m=|\C|$ clauses and with no negative literals, the small modification in constructing the instance $G$ for \SCFVCthree\ is as follows: take $m-1$ additional vertices $c_{12}$, $c_{23}$, \ldots, $c_{m-1,m}$, and make a path from $c_1$ to $c_m$ through the vertices~$c_{12}, c_{23},\ldots,c_{m-1,m}$. See Fig.~\ref{fig:dp}.  

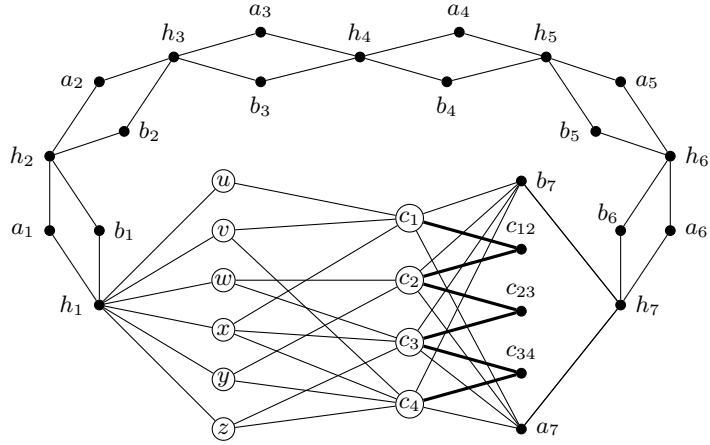
\begin{figure}[ht]
\centering
\begin{tikzpicture}[scale=.33]
\node[vertexV] (u) at (8,11) {\small $u$};
\node[vertexV] (v) at (8,9) {\small $v$};
\node[vertexV] (w) at (8,7) {\small $w$};
\node[vertexV] (x) at (8,5) {\small $x$};
\node[vertexV] (y) at (8,3) {\small $y$};
\node[vertexV] (z) at (8,1) {\small $z$};
\node[vertexC] (c1) at (15.5,9.5) {\small $c_1$};
\node[vertexC] (c2) at (15.5,7) {\small $c_2$};
\node[vertexC] (c3) at (15.5,4.5) {\small $c_3$};
\node[vertexC] (c4) at (15.5,2) {\small $c_4$};

\draw (c1)--(u); \draw (c1)--(v); \draw (c1)--(x);
\draw (c2)--(w); \draw (c2)--(y);
\draw (c3)--(w); \draw (c3)--(x); \draw (c3)--(z);
\draw (c4)--(v); \draw (c4)--(x); \draw (c4)--(y); \draw (c4)--(z);

\node[vertex] (h1) at (3,6) [label=left:{\small $h_1$}]{};
\node[vertex] (h2) at (1,12) [label=left:{\small $h_2$}]{};
\node[vertex] (h3) at (6,16) [label=above:{\small $h_3$}]{};
\node[vertex] (h4) at (13.5,16)  [label=above:{\small $h_4$}]{};
\node[vertex] (h5) at (21,16)  [label=above:{\small $h_5$}]{};
\node[vertex] (h6) at (26,12)  [label=right:{\small $h_6$}]{};
\node[vertex] (h7) at (24,6)  [label=right:{\small $h_7$}]{};
\node[vertex] (a1) at (1,9) [label=left:{\small $a_1$}]{};
\node[vertex] (a2) at (3,15) [label=left:{\small $a_2$}]{};
\node[vertex] (a3) at (9.5,17) [label=above:{\small $a_3$}]{};
\node[vertex] (a4) at (17.5,17) [label=above:{\small $a_4$}]{};
\node[vertex] (a5) at (24,15) [label=right:{\small $a_5$}]{};
\node[vertex] (a6) at (26,9) [label=right:{\small $a_6$}]{};
\node[vertex] (a7) at (20,1) [label=right:{\small $a_{7}$}]{};

\node[vertex] (b1) at (3,9) [label=right:{\small $b_1$}]{};
\node[vertex] (b2) at (4,13) [label=right:{\small $b_2$}]{};
\node[vertex] (b3) at (9.5,15) [label=below:{\small $b_3$}]{};
\node[vertex] (b4) at (17,15) [label=below:{\small $b_4$}]{};
\node[vertex] (b5) at (23,13) [label=left:{\small $b_5$}]{};
\node[vertex] (b6) at (24,9) {}; 
\node at (23.5,9.7) {\small $b_6$};
\node[vertex] (b7) at (20,11) [label=right:{\small $b_7$}]{};

\node[vertex] (c12) at (20,8.25) [label=above:{\small $c_{12}$}]{};
\node[vertex] (c23) at (20,5.75)  [label=above:{\small $c_{23}$}]{};
\node[vertex] (c34) at (20,3.25)  [label=above:{\small $c_{34}$}]{};

\draw (h1)--(u); \draw (h1)--(v); \draw (h1)--(w); \draw (h1)--(x); \draw (h1)--(y); \draw (h1)--(z);

\draw (h7)--(a7); \draw (h7)--(b7); 

\draw (a7)--(c1); \draw (a7)--(c2); \draw (a7)--(c3); \draw (a7)--(c4); 
\draw (b7)--(c1); \draw (b7)--(c2); \draw (b7)--(c3); \draw (b7)--(c4); 

\draw[very thick] (c1)--(c12)--(c2)--(c23)--(c3)--(c34)--(c4);

\draw (h1)--(a1)--(h2)--(a2)--(h3)--(a3)--(h4)--(a4)--(h5)--(a5)--(h6)--(a6)--(h7)--(a7);
\draw (h1)--(b1)--(h2)--(b2)--(h3)--(b3)--(h4)--(b4)--(h5)--(b5)--(h6)--(b6)--(h7)--(b7); 

\end{tikzpicture} 

\caption{Polynomial-parameter transformation from \PNAESAT$[n]$ to \SCFVCthree$[\textsc{dp}]$: the bipartite graph $G$ from the formula $(U,\C)$ over variable set $U=\{u,v,w,x,y,z\}$ and clause set $\C=\{C_1,C_2,C_3,C_4\}$ with $C_1=\{u, v, x\}$, $C_2=\{w,y\}$, $C_3=\{w,x,z\}$ and $C_4= \{v,x,y,z\}$.}\label{fig:dp}
\end{figure}

Then, the modulator is $X = U\cup\{a_i,b_i,h_i\mid 1\le i\le 7\}$: the remainder graph $G-X$ is the path $c_1, c_{12}, c_2,\ldots,c_{m-1,m}, c_m$, and $|X|=n+21$ is polynomial in the input parameter $n=|U|$. 

Note that $G$ remains bipartite and has diameter and radius~8. Moreover, 
$(U,\C)$ is yes for \PNAESAT\ if and only if $G$ is yes for \SCFVCthree: 
For one direction, assume that there is a satisfying NAE-assignment \texttt{b} for $\mathcal{C}$. Then color the vertices of $G$ as follows:
\begin{itemize}
\item color variable vertex $v$ by color~$1$ if $\texttt{b}(v)=\texttt{True}$. Otherwise by color~$0$;
\item color all clause vertices $c$ by color $2$;
\item color all $h_i$ by color $2$;
\item color all $a_i$ by color $0$ and all $b_i$ by color $1$;
\item color all $c_{ij}$ by color~$0$.
\end{itemize}
It is not hard to see that this is a strong conflict-free coloring of $G$. 

For the other direction, assuming $G$ admits a strong cfvc $3$-coloring, the same proof as in Theorem~\ref{thm:npc} can be applied to verify that $(U,\C)$ is yes for \PNAESAT. (The additional vertices $c_{ij}$ have no impact on the proof.) 
\qed


\section{Conclusion}\label{sec:con}
In this paper, we study the strong conflict-free vertex-connection coloring, an interesting variant of graph coloring recently introduced in~\cite{HsiehLLP24}. 
The main result states that, when parameterized by the vertex cover number, \SCFVC\ is parameter-tractable but, unless $\containment$, \SCFVCthree\ does not admit a polynomial kernel, even for bipartite graphs of bounded diameter. As a reviewer commented, this is a remarkable result since the somewhat tougher problem, \SCFVC,  is fpt, and the somewhat easier problem, \SCFVCthree, still has no polynomial kernel.

There are major gaps in our understanding of strong conflict-free vertex-connection colorings. First of all, the complexity of \SCFVC\ restricted to trees is still unsettled~\cite{Li25,LiZZMZJ20,LiW18+,Wu25}.

\begin{problem} Is there a polynomial-time algorithm computing the (strong) conflict-free vertex-connection number of a given tree?
\end{problem}

Next, while we believe that, when restricted to bipartite graphs, \SCFVCk\ remains \NP-complete for any fixed $k\ge4$, we are not able to prove this fact.

\begin{problem} Is \SCFVCk\ \NP-complete for any fixed $k\ge4$ when restricted to bipartite graphs?
\end{problem}

Indeed, we believe that Theorems~\ref{thm:eth} and~\ref{thm:nopoly-vc-kernel} accordingly hold true for \SCFVCk\ for any fixed $k\ge3$.

To understand the role of twins in strong cfvc colorings we may ask how far we can go beyond the vertex cover number when improving Theorem~\ref{thm:vc-fpt}. The very first step would to consider the twin cover number.

\begin{problem} Does there exist an fpt algorithm for \SCFVC\ parameterized by the twin cover number? 
\end{problem}

Recall that a \emph{twin cover} (cf.~\cite{Ganian11,Ganian15}) of a graph~$G$ is a set $X\subseteq V(G)$ such that $N[u]=N[v]$ holds for every edge $\{u,v\}$ of $G-X$, i.e., every two adjacent vertices of $G-X$ are twins (which means that they must be true twins). The \emph{twin cover number} is the smallest size of a twin cover. Note that any vertex cover is a twin cover, hence the twin cover number is at most the vertex cover number.

Finally, although we have shown that, under $\mathsf{NP\nsubseteq coNP/poly}$, \SCFVCthree$[\textup{\textsc{dp}}]$ does not admit a polynomial kernel, we do not know whether the problem has a kernel at all. 

\begin{problem} Does there exist an fpt algorithm for \SCFVCthree$[\textup{\textsc{dp}}]$? More general: does there exist an fpt algorithm for \SCFVC\ parameterized by the distance to a path? 
\end{problem}

\bibliography{DA18126grow24}

\end{document}